\documentclass[
reprint,
superscriptaddress,
groupedaddress,
nofootinbib,
amsmath,amssymb,
aps,
prl,
floatfix,longbibliography
]{revtex4-2}
\usepackage{graphicx}
\usepackage{dcolumn}
\usepackage[hypertexnames, breaklinks=true, bookmarksnumbered=true, bookmarksopen=true, colorlinks=true, linktocpage=true, allcolors=magenta]{hyperref}
\usepackage{amsmath,amssymb,amsfonts,amsthm}
\usepackage{mathtools} 
\usepackage[T1]{fontenc}
\usepackage[latin9]{inputenc}
\usepackage{times}
\usepackage{txfonts}
\usepackage{bbm} 
\usepackage{bm} 
\usepackage{mathrsfs} 
\usepackage{xcolor}
\usepackage{natbib}
\usepackage{braket}

\DeclareMathOperator{\Tr}{Tr}
\allowdisplaybreaks

\newtheorem{theorem}{Theorem}

\newtheorem{proposition}{Proposition}

\begin{document}

\title{Catalytic Transformations of Pure Entangled States}
\author{Tulja Varun Kondra}
\email{t.kondra@cent.uw.edu.pl}

\author{Chandan Datta}

\author{Alexander Streltsov}

\affiliation{Centre for Quantum Optical Technologies, Centre of New Technologies,
University of Warsaw, Banacha 2c, 02-097 Warsaw, Poland}

\begin{abstract}
    Quantum entanglement of pure states is usually quantified via the entanglement entropy, the von Neumann entropy of the reduced state. Entanglement entropy is closely related to entanglement distillation, a process for converting quantum states into singlets, which can then be used for various quantum technological tasks. The relation between entanglement entropy and entanglement distillation has been known only for the asymptotic setting, and the meaning of entanglement entropy in the single-copy regime has so far remained open. Here we close this gap by considering entanglement catalysis. We prove that entanglement entropy completely characterizes state transformations in the presence of entangled catalysts. Our results imply that entanglement entropy quantifies the amount of entanglement available in a bipartite pure state to be used for quantum information processing, giving asymptotic results an operational meaning also in the single-copy setup.
\end{abstract}

\maketitle

Originated in chemistry, catalysis allows one to increase the rate of a chemical reaction. This is achieved by using a catalyst, a substance which is not consumed in the process, and can thus be used repeatedly without additional costs. Similarly, a quantum catalyst is a quantum system which is not changed by the process under consideration, giving access to transformations which are not achievable without it. In the early days of quantum information science, the investigation of catalysis has focused on transformations of entangled states via local operations and classical communication (LOCC)~\cite{NielsenPhysRevLett.83.436,JonathanPhysRevLett.83.3566,EisertPhysRevLett.85.437}.

One of the first examples~\cite{JonathanPhysRevLett.83.3566} demonstrating the power of catalysis in quantum theory involves pure entangled states shared by two parties, Alice and Bob. Two states, denoted by $\ket{\psi}^{AB}$ and $\ket{\phi}^{AB}$, are chosen such that no LOCC procedure can convert $\ket{\psi}^{AB}$ into $\ket{\phi}^{AB}$. Such states can be found using conditions for LOCC transformations presented in~\cite{NielsenPhysRevLett.83.436}. Even if a direct conversion from $\ket{\psi}^{AB}$ to $\ket{\phi}^{AB}$ is not possible, in some cases conversion can still be achieved by using a catalyst. This is an additional quantum system in an entangled state $\ket{\mu}^{A'B'}$, enabling the transformation $\ket{\psi}^{AB} \otimes \ket{\mu}^{A'B'} \rightarrow \ket{\phi}^{AB} \otimes \ket{\mu}^{A'B'}$. Since the state of the catalyst remains unchanged in the process, it can be reused for another transformation in the future. A complete characterization of pure quantum states which can be transformed into each other via LOCC with a catalyst has so far remained open. Partial results addressing this question have been presented over the last decades~\cite{VidalPhysRevA.62.012304,DuanPhysRevA.71.042319,Turgut_2007,Klimesh0709.3680,Aubrun2008,SandersPhysRevA.79.054302,GraboweckyPhysRevA.99.052348}.

Catalysis also enhances transformations in quantum thermodynamics, lifting the well known second law in the classical domain to many second laws in the quantum regime~\cite{Brandao3275}. Moreover, in quantum thermodynamics any quantum state can be used as a universal catalyst~\cite{Lipka-Bartosik2006.16290}. Catalytic properties of quantum coherence~\cite{BuPhysRevA.93.042326,StreltsovRevModPhys.89.041003}, purity~\cite{GOUR20151}, and theories having certain symmetries~\cite{Marvian_2013} have also been considered. There has also been significant interest in correlated catalysts~\cite{AbergPhysRevLett.113.150402,Vaccaro_2018,LostaglioPhysRevLett.123.020403,Wilminge19060241}. Allowing a catalyst to build up correlations with the system has shown to enhance the transformation power of the corresponding procedure~\cite{BoesPhysRevLett.122.210402,wilming2020entropy,shiraishi2020quantum}.

\begin{figure}
\includegraphics[width=0.8\columnwidth]{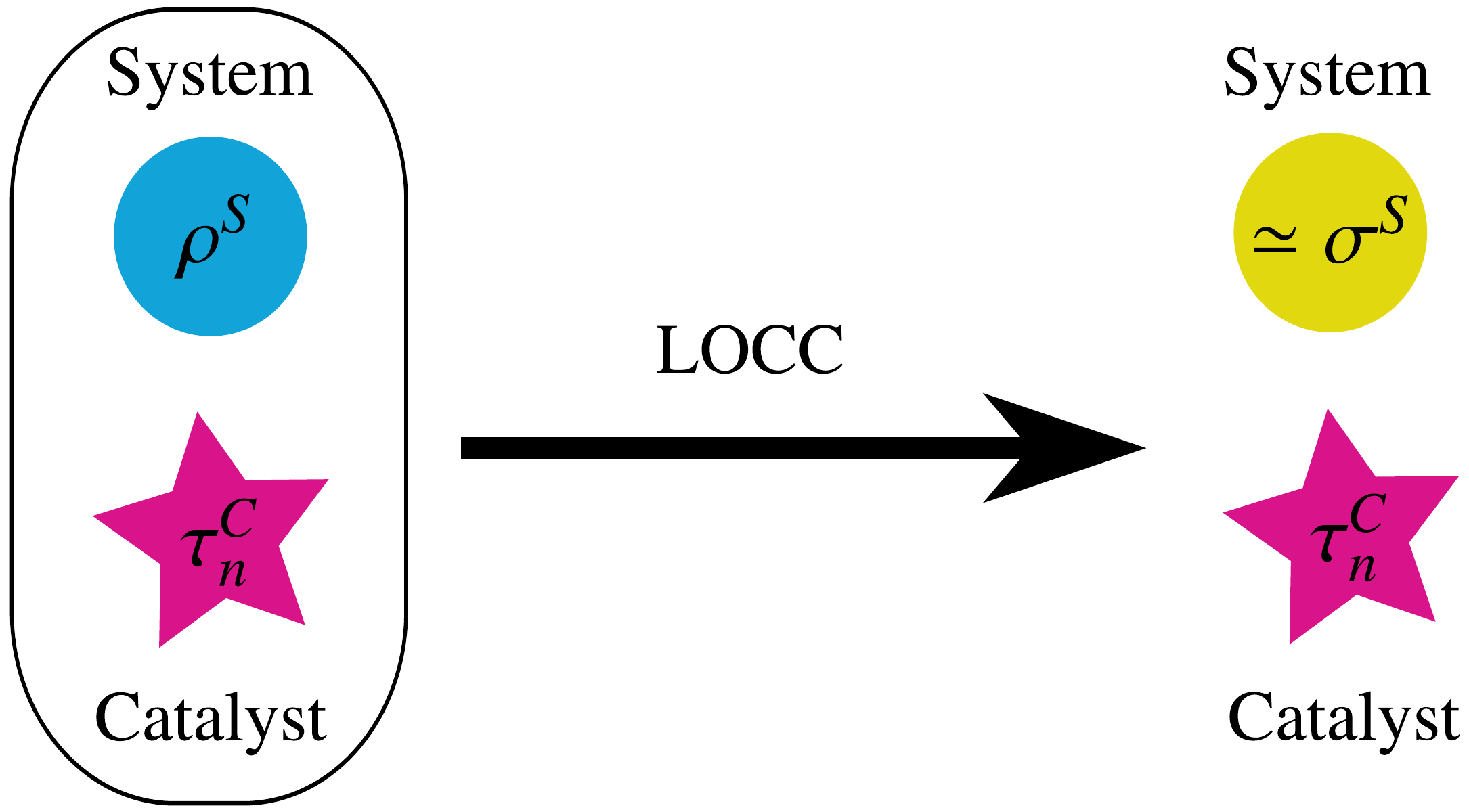}

\caption{Catalytic LOCC transformation from $\rho^S$ to $\sigma^S$ with a sequence of catalyst states $\{\tau_n^C\}$. The state of the catalyst does not change in the procedure, and the system becomes decoupled from the catalyst for $n \rightarrow \infty$. If $\rho^S$ and $\sigma^S$ are bipartite pure states, the transition is fully characterized by entanglement entropy of the states.}
\label{fig:CatalyticLOCC}
\end{figure}

In this Letter we consider \emph{catalytic LOCC transformations}. For a bipartite system $S = AB$, a catalytic LOCC transformation is defined as
\begin{equation}
    \rho^S \rightarrow \lim_{n \rightarrow \infty} \mathrm{Tr}_{C}\left[\Lambda_{n}\left(\rho^S \otimes \tau^C_n\right)\right]. \label{eq:catalytic}
\end{equation}
Here, $C=A'B'$ is a bipartite system of the catalyst, $\left\{\tau^C_n\right\}$ is a sequence of catalyst states, and $\{\Lambda_n\}$ is a sequence of LOCC protocols, see also Fig.~\ref{fig:CatalyticLOCC}. We require that for all $n$ the catalyst remains unchanged in the process: 
\begin{equation}
    \mathrm{Tr}_S\left[\Lambda_{n}\left(\rho^S \otimes \tau^C_n\right)\right] = \tau^C_n. \label{eq:CatalystUnchanged}
\end{equation}
In general, we do not bound the dimension of the catalyst, and we further require that the system decouples from the catalyst for large $n$. In particular, for a catalytic transformation from $\rho^S$ to $\sigma^S$ we require:
\begin{equation}
    \lim_{n\rightarrow \infty} \left|\left|\mu^{SC}_n - \sigma^S\otimes \tau^C_n\right|\right|_1 = 0, \label{eq:DecouplingMethods}
\end{equation}
where $||M||_1 = \Tr\sqrt{M^\dagger M}$ is the trace norm and $\mu^{SC}_n = \Lambda_{n}\left(\rho^S \otimes \tau^C_n\right)$ is the total final state. Eq.~(\ref{eq:CatalystUnchanged}) ensures that the catalyst remains unchanged in the procedure, and thus can be used for another process in the future. Moreover, Eq.~(\ref{eq:DecouplingMethods}) means that the correlations between the system and the catalyst can be made arbitrarily small, implying that multiple uses of the catalyst on independent systems will lead to a negligible amount of correlations between the systems involved.

We will now show that for pure states catalytic LOCC transformations are closely related to asymptotic LOCC transformations. In the asymptotic setting, the parties can operate on a large number of copies of the initial state simultaneously. The figure of merit for the process is the transformation rate, giving the maximal number of copies of $\ket{\phi}^{AB}$ achievable per copy of the initial state $\ket{\psi}^{AB}$. It has been shown in~\cite{BennettPhysRevA.53.2046} that the optimal rate is given by $H\left(\psi^A\right)/H\left(\phi^A\right)$, where $\psi^A$ and $\phi^A$ are the reduced states of $\ket{\psi}^{AB}$ and $\ket{\phi}^{AB}$, and $H(\rho) = - \mathrm{Tr}\left[\rho \log_2 \rho\right]$ is the von Neumann entropy. For pure states the von Neumann entropy of the reduced state is also a quantifier of entanglement, known as entanglement entropy~\cite{BennettPhysRevA.53.2046,VedralPhysRevLett.78.2275,HorodeckiRevModPhys.81.865}: $E\left(\ket{\psi}^{AB}\right) = H\left(\psi^A\right)$.  If the initial and the target state have the same entanglement entropy, then conversion $\ket{\psi}^{AB} \rightarrow \ket{\phi}^{AB}$ is possible with unit rate.

We are now ready to present the first result of this Letter. Alice and Bob can convert $\ket{\psi}^{AB}$ into $\ket{\phi}^{AB}$ via catalytic LOCC if and only if
\begin{equation}
    H\left(\psi^A\right) \geq H\left(\phi^A\right). \label{eq:Pure}
\end{equation}
This result means that for pure states catalysts enhance the transformation power of LOCC, making the transformations as powerful as in the asymptotic limit. The result in Eq.~(\ref{eq:Pure}) is a consequence of the following theorem, concerning transformations between quantum states $\rho^S$ and $\sigma^S$ of a general multipartite system $S$ via multipartite LOCC.

\begin{theorem} \label{thm:Main}
If $\rho^S$ can be transformed into $\sigma^S$ via asymptotic LOCC with unit rate, then there exists a catalytic LOCC protocol transforming $\rho^S$ into $\sigma^S$.
\end{theorem}
\begin{proof}

We consider a system $S$ consisting of $m$ parties and restricted to LOCC. Let all the parties share a state $\rho$. Moreover, we assume that for any $\varepsilon > 0$ there exists an integer $n$ and an LOCC transformation $\Lambda$ such that
\begin{equation}
    \Lambda\left(\rho^{\otimes{n}}\right) = \Gamma \,\,\,\mbox{and}\,\,\,\,D\left(\Gamma,\sigma^{\otimes{n}} \right) < \varepsilon, \label{asymptotic}
\end{equation}
where $D(\rho,\sigma) = \frac{1}{2}||\rho - \sigma||_1$ is the trace distance. The above relation implies that $\sigma$ is asymptotically achievable from $\rho$ with unit rate. We will now show that in this case there also exists a catalytic LOCC procedure transforming $\rho$ into $\sigma$. The following proof is inspired by techniques introduced very recently within quantum thermodynamics~\cite{shiraishi2020quantum}.

Consider a catalyst in the state
\begin{equation}
    \tau = \frac{1}{n}\sum_{k=1}^{n}\rho^{\otimes k-1} \otimes \Gamma_{n-k} \otimes \ket{k}\bra{k}. \label{eq:tau}
\end{equation}
The Hilbert space of the catalyst is in $S^{ \otimes{n-1}}\otimes K$, where $n$ is the integer introduced in Eq. (\ref{asymptotic}) and $K$ represents an auxiliary system of dimension $n$. For brevity, we denote the initial system $S$ as $S_{1}$, and $n-1$ copies of the same system which belong to the catalyst are denoted by $S_{2},\ldots , S_{n}$. Thus, the system $C$ of the catalyst is identified with $C = S_2 \ldots S_n K$. Moreover, $\Gamma$ is a quantum state on $S_{1}\otimes S_{2}\otimes \cdots \otimes S_{n}$, see also Eq.~(\ref{asymptotic}), and $\Gamma_{i}$ is the reduced state of $\Gamma$ on $S_{1}\otimes S_{2}\otimes \cdots\otimes S_{i}$. We further define $\Gamma_{0}=1$. The auxiliary system $K$ is maintained by Alice, serving as a register with a Hilbert space of dimension $n$ with basis $\{\ket{k}, k \in [1,n]\}$. 

Consider now the following LOCC protocol acting on the system and the catalyst:

(i) Alice performs a rank-1 projective measurement on the auxiliary system $K$ in the basis $\ket{k}$. She then communicates the outcome of the measurement to all the other parties. If Alice obtains the outcome $n$, all parties perform the LOCC protocol $\Lambda$ given in Eq. (\ref{asymptotic}) on $S_{1}\otimes S_{2}\otimes\cdots\otimes S_{n}$. For any other outcome of Alice's measurement the parties do nothing.

(ii) Alice applies a unitary on the auxiliary system that converts $\ket{n}\xrightarrow{}\ket{1}$ and $\ket{i}\xrightarrow{}\ket{i+1}$.

(iii) Finally, all the parties apply a SWAP unitary on their parts of ($S_{i}$, $S_{i+1}$) and  ($S_{1}$, $S_{n}$), which shifts $S_{i}\xrightarrow{} S_{i+1}$ and $S_{n}\xrightarrow{} S_{1}$.

The initial state of the system and the catalyst is given by
\begin{equation}
    \rho \otimes \tau = \frac{1}{n}\sum_{k=1}^{n}\rho^{\otimes k} \otimes \Gamma_{n-k} \otimes \ket{k}\bra{k}.
\end{equation}
After applying step (i), the initial state transforms to
\begin{equation}
    \mu^{i} = \frac{1}{n}\sum_{k=1}^{n-1}\rho^{\otimes k} \otimes \Gamma_{n-k} \otimes \ket{k}\bra{k} + \frac{1}{n}\Gamma\otimes \ket{n}\bra{n}.
\end{equation}
In step (ii), $\mu^{i}$ transforms to $\mu^{ii}$, where 
\begin{equation}
    \mu^{ii} = \frac{1}{n}\sum_{k=1}^{n}\rho^{\otimes k-1} \otimes \Gamma_{n+1-k} \otimes \ket{k}\bra{k}. \label{eq:Muii}
\end{equation}
Note that tracing out $S_{n}$ from $\mu^{ii}$ gives $\tau$, which is the initial state of the catalyst, see Eq.~(\ref{eq:tau}). Therefore, using step (iii), we transform $\mu^{ii}$ to the final state $\mu$ having the property $\Tr_{S}[\mu]$ = $\tau$. This proves that the state of the catalyst does not change in this procedure.

To complete the proof, we will now show that for any $\varepsilon >0$ there is an integer $n$ such that 
\begin{equation}
    D\left(\mu^{SC}, \sigma^S \otimes \tau^C\right) < 2\varepsilon, \label{eq:DecouplingMain}
\end{equation}
where $\mu^{SC}$ is the total final state given above in this proof. For this, recall that $\mu^{SC}$ is equivalent to the state $\mu^{ii}$ in Eq.~(\ref{eq:Muii}) up to a cyclic SWAP. This implies that our figure of merit $D\left(\mu^{SC},\sigma^S \otimes \tau^C\right)$ is equal to $D\left(\mu^{ii},\nu\right)$, where the state $\nu$ is given as
\begin{equation}
    \nu = \frac{1}{n}\sum_{k=1}^{n}\rho^{\otimes k-1} \otimes \Gamma_{n-k} \otimes \sigma \otimes \ket{k}\bra{k}.
\end{equation}
We further obtain 
\begin{align}
    D\left(\mu^{ii} , \nu\right) &= \frac{1}{n} \sum_{k=1}^{n} D\left(\Gamma_{n+1-k}, \Gamma_{n-k} \otimes \sigma\right) \\
                      &\leq \frac{1}{n}\sum_{k=1}^{n}D\left(\Gamma_{n+1-k}, \sigma^{\otimes(n+1-k)}\right)  \nonumber \\ 
                      & + \frac{1}{n}\sum_{k=1}^{n}D\left(\Gamma_{n-k}\otimes \sigma, \sigma^{\otimes(n+1-k)}\right) < 2 \varepsilon, \nonumber
\end{align}
where in the first inequality we used the triangle inequality and in the second inequality we used Eq.~(\ref{asymptotic}) together with the monotonicity of trace distance under partial trace. This proves Eq.~(\ref{eq:DecouplingMain}), showing also that the system and the catalyst decouple in the procedure, and that Eq.~(\ref{eq:DecouplingMethods}) is fulfilled. This completes the proof of the theorem. We also note that a family of catalyst states similar to Eq.~(\ref{eq:tau}) has been previously considered in~\cite{Lipka-Bartosik2006.16290}.
\end{proof}

In the bipartite setting, Theorem~\ref{thm:Main} directly implies that $\ket{\psi}^{AB}$ can be transformed into $\ket{\phi}^{AB}$ via catalytic LOCC if Eq.~(\ref{eq:Pure}) is fulfilled. As we will see in the following by using properties of entanglement quantifiers~\cite{Christandl_2004,Brandao2011}, a transformation is not possible if Eq.~(\ref{eq:Pure}) is violated. 

For this, we will show that the squashed entanglement, introduced in \cite{Christandl_2004}, is monotonic under catalytic LOCC transformations. For bipartite quantum states $\rho^{AB}$ the squashed entanglement is defined as~\cite{Christandl_2004}
\begin{equation}
    E_{sq}\left(\rho^{AB}\right) = \text{inf}\left\{\frac{1}{2}I(A;B|E): \rho^{ABE}\,\,\text{extension of}\,\, \rho^{AB} \right\},
\end{equation}
where the infimum is taken over all quantum states $\rho^{ABE}$ with $\rho^{AB}$ = $\Tr_{E} \left(\rho^{ABE}\right)$ and 
\begin{equation}
    I(A; B|E) = H\left(\rho^{AE}\right) + H\left(\rho^{BE}\right) - H\left(\rho^{ABE}\right) - H\left(\rho^{E}\right)
\end{equation}
is the quantum conditional mutual information of $\rho^{ABE}$. 

We use the following properties of the squashed entanglement~\cite{Christandl_2004}: (a) $E_{sq}$ is an entanglement monotone, i.e. it does not increase under LOCC. (b) $E_{sq}$ is superadditive in general and additive on tensor products:
\begin{equation}
    E_{sq}\left(\rho^{AA'BB'}\right) \geq E_{sq}\left(\rho^{AB}\right) +E_{sq}\left(\rho^{A'B'}\right)
\end{equation}
and equality holds true if $\rho^{AA'BB'} = \rho^{AB} \otimes \rho^{A'B'}$. (c) For a pure state $\ket{\psi}^{AB}$ squashed entanglement is equal to the entanglement entropy, i.e., the entropy of the reduced state: $E_{sq}\left(\ket{\psi}^{AB}\right) = H\left(\psi^{A}\right)$. We are now ready to prove that the squashed entanglement is monotonic under catalytic LOCC.
\begin{theorem} \label{thm:SquashedEntanglement}
If a bipartite state $\rho^{AB}$ can be transformed into another state $\nu^{AB}$ via catalytic LOCC, then 
\begin{equation}
    E_{sq}\left(\rho^{AB}\right) \geq E_{sq}\left(\nu^{AB}\right).
\end{equation}
\end{theorem}
\begin{proof}
Assume that for any $\varepsilon > 0$ there exists a catalyst state $\tau^{A'B'}$ and an LOCC protocol $\Lambda$ such that the final state $\sigma^{AA'BB'}=\Lambda(\rho^{AB}\otimes\tau^{A'B'})$ has the properties 
\begin{align}
D\left(\Tr_{A'B'}\left[\sigma^{AA'BB'}\right],\nu^{AB}\right)<\varepsilon, \\
\Tr_{AB}\left[\sigma^{AA'BB'}\right]=\tau^{A'B'}.
\end{align}
Using the properties (a) and (b) of the squashed entanglement, we find
\begin{equation}
    E_{sq}\left(\sigma^{AA'BB'}\right)\leq E_{sq}\left(\rho^{AB}\right)+E_{sq}\left(\tau^{A'B'}\right) \label{1ineq}
\end{equation}
and also
\begin{equation}
    E_{sq}\left(\sigma^{AA'BB'}\right)\geq E_{sq}\left(\Tr_{A'B'}\left[\sigma^{AA'BB'}\right]\right)+E_{sq}\left(\tau^{A'B'}\right). \label{2ineq}
\end{equation}
From Eqs. (\ref{1ineq}) and (\ref{2ineq}) it follows
\begin{equation}
    E_{sq}\left(\rho^{AB}\right)\geq E_{sq}\left(\Tr_{A'B'}\left[\sigma^{AA'BB'}\right]\right).
\end{equation}
If $\Tr_{A'B'}\left[\sigma^{AA'BB'}\right]$ can be made arbitrarily close to $\nu^{AB}$ in trace distance, then by continuity of squashed entanglement~\cite{Alicki_2004} we get
$E_{sq}\left(\rho^{AB}\right) \geq  E_{sq}\left(\nu^{AB}\right)$, and the proof is complete.
\end{proof}
Combining Theorems~\ref{thm:Main} and \ref{thm:SquashedEntanglement}, we conclude that a pure state $\ket{\psi}^{AB}$ can be transformed into another pure state $\ket{\phi}^{AB}$ via catalytic LOCC if and only if $H\left(\psi^A\right) \geq H\left(\phi^A\right)$, as claimed in Eq.~(\ref{eq:Pure}). It is instrumental to compare our results to the results in~\cite{Turgut_2007,Klimesh0709.3680}, where exact catalytic transformations between pure states have been studied, using a catalyst which is also in a pure state. If no error is allowed at any stage of the protocol, the transitions are characterized by a set of entropic inequalities~\cite{Turgut_2007,Klimesh0709.3680}. In this Letter -- in contrast to~\cite{Turgut_2007,Klimesh0709.3680} -- we allow for a catalyst in a mixed state. While the state of the catalyst is not changed in the process, we also allow for an error in the state of the system, requiring that the error vanishes in the limit of large catalysts. We see that catalytic transformations between bipartite pure states are fully characterized by the entanglement entropy. This gives an operational interpretation for the von Neumann entropy in the single copy scenario.

Theorem~\ref{thm:Main} has surprising implications for general bipartite states, when it comes to catalytic entanglement distillation and dilution~\cite{BennettPhysRevA.53.2046,BennettPhysRevLett.76.722}. In particular, via catalytic LOCC it is possible to convert any distillable state $\rho^{AB}$ into an entangled pure state $\ket{\psi}^{AB}$ having entanglement entropy $E\left(\ket{\psi}^{AB}\right) = E_\mathrm{d}\left(\rho^{AB}\right)$, where $E_\mathrm{d}$ is the distillable entanglement~\cite{PlenioMeasures,HorodeckiRevModPhys.81.865}. On the other hand, via catalytic LOCC it is possible to create any quantum state $\rho^{AB}$ from a pure state $\ket{\psi}^{AB}$ with entanglement entropy $E\left(\ket{\psi}^{AB}\right) = E_\mathrm{c}\left(\rho^{AB}\right)$, where $E_\mathrm{c}$ is the entanglement cost~\cite{PlenioMeasures,HorodeckiRevModPhys.81.865}.

\begin{figure*}
\includegraphics[width=1.7\columnwidth]{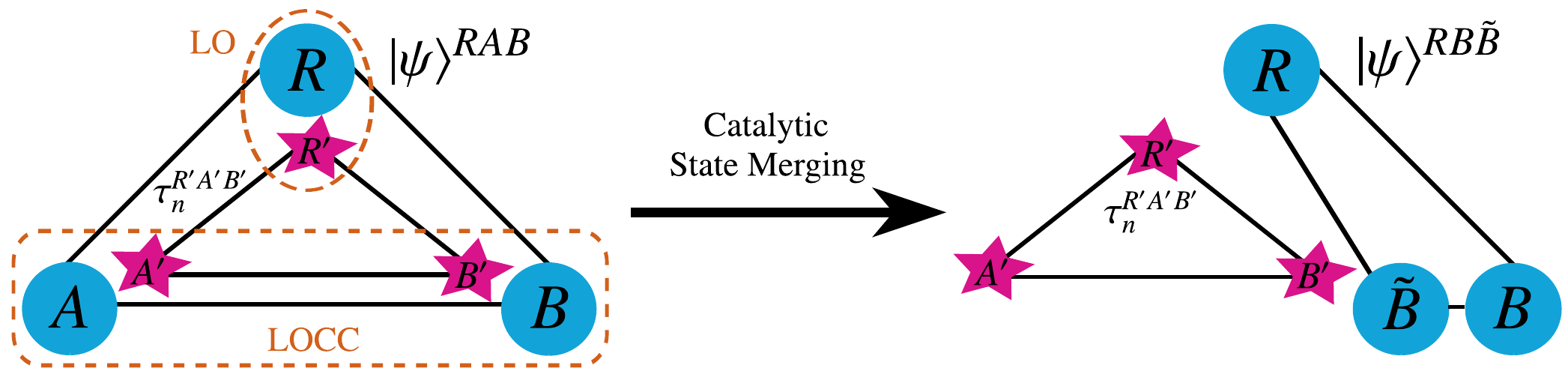}

\caption{Catalytic quantum state merging. Alice, Bob, and Referee share a single copy of $\ket{\psi}^{RAB}$. Alice aims to send her part of the state to Bob by using catalytic LOCC, and the Referee can apply local unitary transformations (LO). The process is completely characterized by the quantum conditional entropy $H(A|B)$, see the main text for more details.}
\label{fig:Merging}
\end{figure*}

Remarkably, Theorem~\ref{thm:Main} holds not only for bipartite state transformations, but also for multipartite LOCC protocols. Here the goal is to convert a multipartite state $\rho^S$ into another state $\sigma^S$ via multipartite LOCC. As a consequence, it allows us to translate a broad range of asymptotic results in entanglement theory to a corresponding result on the single-copy level. We show it explicitly for a variation of quantum state merging, which we term \emph{catalytic quantum state merging}. Before we present this task, we review the standard quantum state merging procedure in the following.

In quantum state merging~\cite{qsm,QSM_and_negative_info}, Alice, Bob, and Referee share asymptotically many copies of a quantum state $\ket{\psi}^{RAB}$. The goal of the process is to send Alice's part of the state to Bob, while preserving correlations with the Referee. Alice and Bob can perform LOCC protocols and share additional singlets. As was shown in~\cite{qsm,QSM_and_negative_info}, the performance of this process is characterized by the quantum conditional entropy
\begin{equation}
    H(A|B) = H\left(\psi^{AB}\right) - H\left(\psi^B\right).
\end{equation}
For $H(A|B) > 0$ quantum state merging can be performed if Alice and Bob share additional singlets at rate $H(A|B)$, and merging is not possible if less singlets are available. For $H(A|B) \leq 0$ Alice and Bob can perform quantum state merging with LOCC, while additionally gaining singlets at rate $-H(A|B)$. Remarkably, quantum state merging gives an operational meaning to the quantum conditional entropy, regardless whether $H(A|B)$ is positive or negative.

We are now ready to define catalytic quantum state merging, giving the quantum conditional entropy an operational meaning also in the single-copy regime. Here, Alice, Bob, and Referee share one copy of the state $\ket{\psi}^{RAB}$, and can use additional catalysts in arbitrary states $\tau_n^{R'A'B'}$. While in standard quantum state merging the Referee is fully inactive, in catalytic quantum state merging we allow the Referee to perform local unitaries. However, communication between the Referee and the other parties is not required, see also Fig.~\ref{fig:Merging}. The goal is to merge the single copy of $\ket{\psi}^{RAB}$ on Bob's side without changing the state of the catalyst for all $n$, and with decoupling of the catalyst in the limit $n\rightarrow \infty$. We find that for $H(A|B) > 0$ catalytic quantum state merging can be performed if Alice and Bob additionally share a pure entangled state with entanglement entropy $H(A|B)$. This procedure is optimal: merging is not possible if a pure state with a smaller entanglement entropy is provided. If $H(A|B) \leq 0$, then catalytic state merging can be performed without extra entanglement. In the end of the process, Alice and Bob can gain an additional pure state with entanglement entropy $-H(A|B)$. Also this procedure is optimal: it is not possible to achieve merging and gain a pure state with entanglement entropy exceeding~$-H(A|B)$. We refer to the Supplemental Material for the proofs and more details.

As a final example we discuss assisted entanglement distillation~\cite{DiVincenzo10.1007/3-540-49208-9_21,SmolinPhysRevA.72.052317}, where three parties, Alice, Bob, and Charlie, share a pure state $\ket{\psi}^{ABC}$. By performing LOCC involving all parties, their aim is to extract singlets between Alice and Bob. In the asymptotic setup, the optimal singlet rate is given by $\min\left\{H\left(\psi^A\right),H\left(\psi^B\right)\right\}$~\cite{SmolinPhysRevA.72.052317}. Correspondingly, \emph{catalytic assisted entanglement distillation} involves one copy of $\ket{\psi}^{ABC}$. By applying catalytic LOCC, the parties aim to establish a state $\ket{\phi}^{AB}$ shared by Alice and Bob, having entanglement entropy as large as possible. We find that $\min\left\{H\left(\psi^A\right),H\left(\psi^B\right)\right\}$ corresponds to the maximal entanglement entropy achievable from $\ket{\psi}^{ABC}$ in this procedure. The technical details can be found in the Supplemental Material.

In summary, we have shown that entanglement entropy plays a fundamental role in quantifying the amount of entanglement available in a bipartite pure state to be used for quantum information processing. This holds in the single-copy setup, where a single instance of a pure quantum state is available, supported by a suitable catalyst which remains unchanged in the process.  Our methods are directly applicable to studying the role of catalysis in more general quantum information tasks. We have demonstrated this explicitly for quantum state merging, and assisted entanglement distillation. It is reasonable to assume that similar results will hold for other quantum information protocols, including also the recently developed procedures for entangled state transformations in multipartite setups~\cite{StreltsovPhysRevLett.125.080502}. 

For bipartite pure entangled states, our methods allow to construct catalyst states which achieve optimal performance, as described by Eq.~(\ref{eq:Pure}). These catalyst states are tailored to the specific state transition, allowing us to implement the transformation repeatedly with an arbitrary small error.
An interesting open question is whether optimal performance can also be achieved with a universal catalyst. Such possibility has been pointed out in~\cite{Lipka-Bartosik2006.16290} for quantum thermodynamics and related theories. Extension of these tools to the setup introduced in this Letter is beyond the scope of our work, and we leave it open for future research. Our results suggest a full equivalence between asymptotic and catalytic entanglement theory.

\emph{Note added}: The authors of~\cite{Lipka-Bartosik2102.11846} have used similar techniques to study the role of entangled catalysts for quantum teleportation, showing that catalysts can significantly enhance the teleportation fidelity.

We thank Patryk Lipka-Bartosik and Paul Skrzypczyk for discussion and insightful comments on our manuscript. This work was supported by the ``Quantum Optical Technologies'' project, carried out within the International Research Agendas programme of the Foundation for Polish Science co-financed by the European Union under the European Regional Development Fund.

\bibliography{literature}

\clearpage

\section{Supplemental Material}

\subsection{Catalytic LOCC protocols in tripartite setups}
We will now consider a tripartite setup, developing tools which will serve as a basis for catalytic quantum state merging. Similar to the state merging setup~\cite{qsm, QSM_and_negative_info}, we consider three parties (Alice, Bob and Referee) sharing a tripartite state $\rho = \rho^{RAB}$. Assume now that by applying asymptotic LOCC between Alice and Bob, it is possible to asymptotically convert $\rho$ into another tripartite state $\sigma = \sigma^{RAB}$. The following theorem establishes a connection between this setup and catalytic LOCC. 

\begin{proposition}
If  $\rho$ can be converted into  $\sigma$ via asymptotic LOCC between Alice and Bob with unit rate, then $\rho$ can be converted into $\sigma$ by applying catalytic LOCC between Alice and Bob and a unitary on Referee's side. \label{merging theorem}
\end{proposition}
\begin{proof}
The proof follows similar reasoning as the proof of Theorem~1 in the main text. If  $\rho$ can be converted into  $\sigma$ via asymptotic LOCC between Alice and Bob with unit rate, then for any $\varepsilon > 0$ there exists an integer $n$ and an LOCC protocol $\Lambda$ between Alice and Bob such that 
\begin{equation}
    \Lambda\left[\rho^{\otimes n}\right]=\Gamma\,\,\,\,\,\mathrm{and}\,\,\,\,D\left(\Gamma,\sigma^{\otimes n}\right)<\varepsilon. \label{asymptotic2}
\end{equation}
We now consider a catalyst in the state 
\begin{equation}
    \tau = \frac{1}{n}\sum_{k=1}^{n}\rho^{\otimes k-1} \otimes \Gamma_{n-k} \otimes \ket{k}\bra{k}.
\end{equation}
Also in this case the Hilbert space of the catalyst is in $S^{ \otimes{n-1}}\otimes K$, where $S = RAB$ now corresponds to the tripartite system of Alice, Bob and Referee. Again, we denote the $n$ copies of the same systems as $S_{1},\ldots , S_{n}$, where $S_1$ corresponds to the system $S$, and the state of the catalyst is acting on $S_2\otimes \cdots \otimes S_n \otimes K$. The operator $\Gamma$ acts on $S_{1}\otimes S_{2}\otimes \cdots \otimes S_{n}$ and $\Gamma_{i}$ ($i \in \{1,\ldots,n\}$) is the reduced state of $\Gamma$ on $S_{1}\otimes S_{2}\otimes \cdots\otimes S_{i}$. Moreover, $K$ is a register on Alice's side.

We now follow a procedure very similar to the one in the proof of Theorem~1 in the main text.

(i) Alice performs a rank-1 projective measurement on the register $K$ in the basis $\ket{k}$. She then communicates the outcome of the measurement to Bob. If Alice obtains the outcome $n$, Alice and Bob perform the LOCC protocol $\Lambda$ given in Eq. (\ref{asymptotic2}). For any other outcome of Alice's measurement the parties do nothing.

(ii) Alice applies a unitary on the auxiliary system that converts $\ket{n}\xrightarrow{}\ket{1}$ and $\ket{i}\xrightarrow{}\ket{i+1}$.

(iii) Alice, Bob and Referee apply a SWAP unitary on their parts of ($S_{i}$, $S_{i+1}$) and  ($S_{1}$, $S_{n}$) that shift $S_{i}\xrightarrow{} S_{i+1}$ and $S_{n}\xrightarrow{} S_{1}$. Note that communication with the Referee is not necessary, the Referee applies the SWAP operations independently of the procedure performed by Alice and Bob.

By the same reasoning as in the proof of Theorem~1, we see that the first subsystem $S_1$ of the final state is $\varepsilon$ close to $\sigma$ while the catalyst remains unchanged. Moreover, the subsystem $S_1$ decouples from the catalyst in the limit $n \rightarrow \infty$, which is proven exactly in the same way as in Theorem~1, see also the discussion below Eq.~(10) in the main text.
\end{proof}

\subsection{Catalytic quantum state merging}
In quantum state merging (QSM) \cite{qsm, QSM_and_negative_info}, we assume that Alice, Bob and Referee share asymptotically many copies of a pure quantum state $\ket{\psi}^{RAB}$. By applying LOCC operations, Alice and Bob aim to transfer the state of Alice to Bob while preserving correlations with Referee, i.e., the final state $\ket{\psi}^{RBB'}$ is the same as $\ket{\psi}^{RAB}$ up to relabelling of $A$ and $B'$. In the following, we describe three possible scenarios. 

a) In the asymptotic limit where many copies of the state $\ket{\psi}^{RAB}$ are available, QSM is possible if the conditional entropy is zero \cite{qsm, QSM_and_negative_info}, i.e.,
\begin{equation}
    H_{\psi}\left(A|B\right) = H\left(\psi^{AB}\right) - H\left(\psi^{B}\right)=0.
\end{equation}
This means, if $H_{\psi}(A|B) = 0$, there exists an LOCC protocol taking $\ket{\psi}^{RAB}$ arbitrarily close to $\ket{\psi}^{RBB'}$
\begin{equation}
   \left(\ket{\psi}^{RAB}\right)^{\otimes n}\xrightarrow[LOCC]{\varepsilon}\left(\ket{\psi}^{RBB'}\right)^{\otimes n}, \label{ASYM-MERGING1}
\end{equation}
where $\varepsilon$ above the arrow represents that the final state is  $\varepsilon$ close in trace distance to $\left(\ket{\psi}^{RBB'}\right)^{\otimes n}$.

b) If $H_{\psi}\left(A|B\right)>0$, then we define $\ket{\psi'}^{RA\tilde A B\tilde B}$ = $\ket{\psi}^{RAB} \otimes \ket{\phi_1}^{\tilde{A}\tilde{B}}$, where $\ket{\phi_1}^{\tilde{A}\tilde{B}}$ is a shared entangled state between Alice and Bob with entanglement entropy $E\left(\ket{\phi_1}^{\tilde{A}\tilde{B}}\right) = H_{\psi}\left(A|B\right)$. This implies, $H_{\psi'}\left(A\tilde A|B \tilde B\right)$ = 0. Therefore, from Eq. (\ref{ASYM-MERGING1}) we see that Alice and Bob can successfully merge the state $\ket{\psi'}^{RA \tilde A B \tilde B}$:
\begin{equation}
 \left(\ket{\psi'}^{RA \tilde A B \tilde B}\right)^{\otimes n}\xrightarrow[LOCC]{\varepsilon}\left(\ket{\psi}^{RBB'}\right)^{\otimes n}. \label{asym-merging2}
\end{equation}

c) If $H_{\psi}\left(A|B\right)<0$, the following transformation is achievable via LOCC between Alice and Bob~\cite{qsm, QSM_and_negative_info}:
\begin{equation}
    \left(\ket{\psi}^{RAB}\right)^{\otimes n}\xrightarrow[LOCC]{\varepsilon}\left(\ket{\psi}^{RBB'}\right)^{\otimes n}\otimes\left(\ket{\phi^{+}}^{\tilde{A}\tilde{B}}\right)^{\otimes-H_{\psi}\left(A|B\right)n}, \label{asym-merging3}
\end{equation}
with the Bell state $\ket{\phi^{+}}=(\ket{0}\ket{0}+\ket{1}\ket{1})/\sqrt{2}$. Additionally, we know that~\cite{BennettPhysRevA.53.2046} 
\begin{equation}
  \left(\ket{\phi^{+}}^{\tilde{A}\tilde{B}}\right)^{\otimes-H_{\psi}\left(A|B\right)n}\xrightarrow[LOCC]{\varepsilon}\left(\ket{\phi_{2}}^{\tilde{A}\tilde{B}}\right)^{\otimes n},\label{asym-merging4}
\end{equation}
where $\ket{\phi_{2}}^{\tilde{A}\tilde{B}}$ is a bipartite pure state with entanglement entropy $E\left(\ket{\phi_{2}}^{\tilde{A}\tilde{B}}\right) = - H_{\psi}\left(A|B\right)$. Therefore, from Eqs. (\ref{asym-merging3}) and~(\ref{asym-merging4}) we see that the following transformation is achievable via LOCC between Alice and Bob:
\begin{equation}
   \left(\ket{\psi}^{RAB}\right)^{\otimes n}\xrightarrow[LOCC]{\varepsilon}\left(\ket{\psi}^{RBB'}\otimes\ket{\phi_{2}}^{\tilde{A}\tilde{B}}\right)^{\otimes n}. \label{asym-merging5}
\end{equation}

Equipped with these tools we are now ready to discuss catalytic quantum state merging. Recall that in catalytic QSM, we allow LOCC operations between Alice and Bob, and the Referee can also perform local unitaries. However, no communication between the Referee and the other parties is allowed. This is exactly the setup considered in Proposition~\ref{merging theorem}, and we will make use of this result in the following. 

An immediate consequence of Proposition~\ref{merging theorem} and Eq.~(\ref{asym-merging5}) is that for 
$H_\psi(A|B) < 0$ catalytic state merging from Alice to Bob is possible. Additionally, Alice and Bob can obtain an entangled state $\ket{\phi_{2}}^{\tilde{A}\tilde{B}}$ with entanglement entropy
$E(\ket{\phi_2}^{\tilde{A}\tilde{B}}) =  -H_\psi(A|B)$.
 
In the following, we prove by contradiction that this is the optimal value one can achieve. Let us assume that there exists a catalytic LOCC procedure such that
\begin{equation}
    E\left(\ket{\phi_{2}}^{\tilde{A}\tilde{B}}\right)> H\left(\psi^{B}\right)-H\left(\psi^{AB}\right). \label{ineq1}
\end{equation}
Consider now the squashed entanglement~\cite{Christandl_2004,Brandao2011} between Bob and the rest of the system in the initial state $\ket{\psi}^{RAB}$:
\begin{equation}
    E_{sq}^{B|AR}\left(\ket{\psi}^{RAB}\right) = H\left(\psi^{B}\right). \label{initial_ent1}
\end{equation}
On the other hand, the squashed entanglement between Bob and the rest of the system in the target state $\ket{\psi}^{RBB'}\otimes\ket{\phi_{2}}^{\tilde{A}\tilde{B}}$ is given by
\begin{equation}
  E_{sq}^{BB'\tilde{B}|\tilde{A}R} \left(\ket{\psi}^{RBB'}\otimes\ket{\phi_{2}}^{\tilde{A}\tilde{B}}\right) = H\left(\psi^{AB}\right)+E\left(\ket{\phi_{2}}^{\tilde{A}\tilde{B}}\right). \label{final_ent1}
\end{equation}
From Eqs. (\ref{ineq1}), (\ref{initial_ent1}), and (\ref{final_ent1}), we obtain
\begin{equation}
     E_{sq}^{BB'\tilde{B}|\tilde{A}R} \left(\ket{\psi}^{RBB'}\otimes\ket{\phi_{2}}^{\tilde{A}\tilde{B}}\right) >  H\left(\psi^{B}\right) = E_{sq}^{B|AR}\left(\ket{\psi}^{RAB}\right).
\end{equation}
This means that the squashed entanglement has increased in the process. This is a contradiction to Theorem~2 in the main text, showing that for pure states squashed entanglement cannot increase under catalytic LOCC. 

In the remaining case $H_\psi(A|B) \geq 0$, from Proposition \ref{merging theorem} and Eq.~(\ref{asym-merging2}) it directly follows that catalytic QSM is possible when Alice and Bob are provided with an additional state $\ket{\phi_{1}}^{\tilde{A}\tilde{B}}$ with entanglement entropy 
$E\left(\ket{\phi_{1}}^{\tilde{A}\tilde{B}}\right) = H_\psi(A|B)$.
We now show that this is the minimal entanglement entropy needed to perform catalytic QSM. Again we use the properties of the squashed entanglement to prove this. The squashed entanglement between Bob and the other parties in the initial state $\ket{\psi}^{RAB} \otimes \ket{\phi_1}^{\tilde{A}\tilde{B}}$ is given by
\begin{equation}
    E_{sq}^{B\tilde B|A\tilde A R}\left(\ket{\psi}^{RAB} \otimes \ket{\phi_1}^{\tilde{A}\tilde{B}} \right) = H\left(\psi^{B}\right) + E\left(\ket{\phi_{1}}^{\tilde{A}\tilde{B}}\right).\label{initial_ent2}
\end{equation}
For the target state $\ket{\psi}^{RBB'}$ we obtain
\begin{equation}
  E_{sq}^{BB'|R} \left(\ket{\psi}^{RBB'} \right) = H\left(\psi^{AB}\right). \label{final_ent2}
\end{equation}
Using again the fact that squashed entanglement cannot increase under catalytic LOCC, we have
\begin{equation}
    E_{sq}^{B\tilde B|A\tilde A R}\left(\ket{\psi}^{RAB} \otimes \ket{\phi_1}^{\tilde{A}\tilde{B}} \right) \geq E_{sq}^{BB'|R} \left(\ket{\psi}^{RBB'} \right). \label{non_increase}
\end{equation}
Hence, from Eqs. (\ref{initial_ent2}), (\ref{final_ent2}) and (\ref{non_increase}), we get an achievable lower bound on the entanglement entropy of $\ket{\phi_1}^{\tilde{A} \tilde{B}}$:
\begin{align}
     E\left(\ket{\phi_1}^{\tilde{A} \tilde{B}}\right) \geq  H\left(\psi^{AB}\right)-H\left(\psi^{B}\right).
\end{align}
\bigskip

\subsection{Catalytic assisted entanglement distillation}
Consider now three parties, Alice, Bob, and Charlie, sharing a pure state $\ket{\psi}^{ABC}$. By performing catalytic LOCC between all the parties, they aim to convert $\ket{\psi}^{ABC}$ into a state $\ket{\phi}^{AB}$ which has maximal possible entanglement entropy. This task is analogous to assisted entanglement distillation, which has been previously studied in the asymptotic setting~\cite{SmolinPhysRevA.72.052317}. 

We will now show that the optimal procedure is for Charlie to merge his state either with Alice or with Bob. For this, consider the corresponding conditional mutual information
\begin{subequations}
\begin{align}
    H_\psi(C|A) &= H\left(\psi^{AC}\right) - H\left(\psi^A\right) = H\left(\psi^B\right) - H\left(\psi^{A}\right), \\
    H_\psi(C|B) &= H\left(\psi^{BC}\right) - H\left(\psi^B\right) = H\left(\psi^A\right) - H\left(\psi^B\right).
\end{align}
\end{subequations}
We immediately see that either $H_\psi(C|A)$ and $H_\psi(C|B)$ are both zero, or at least one of them is negative. 

If $H_\psi(C|A) < 0$, then Charlie merges his system with Alice by using catalytic QSM. As a result, Alice and Bob will end up with a state having entanglement entropy $H\left(\psi^B\right)$. Note that $H_\psi(C|A) < 0$ is equivalent to $H\left(\psi^B\right) < H\left(\psi^A\right)$. On the other hand, if $H_\psi(C|B) \leq 0$, then Charlie merges his system with Bob, leaving Alice and Bob with a state having entanglement entropy $H\left(\psi^A\right)$. Since $H_\psi(C|B) \leq 0$ is equivalent to $H\left(\psi^A\right) \leq H\left(\psi^B\right)$, this proves that via catalytic LOCC it is possible to convert $\ket{\psi}^{RAB}$ into a quantum state $\ket{\phi}^{AB}$ having entanglement entropy
\begin{equation}
    E\left(\ket{\phi}^{AB}\right) = \min\left\{H\left(\psi^A\right),H\left(\psi^B\right)\right\}.
\end{equation}
The converse can be proven by using the properties of the squashed entanglement~\cite{Christandl_2004}, in particular that for pure states it corresponds to the entanglement entropy and does not increase under catalytic LOCC, see Theorem~2 in the main text. Since any tripartite LOCC protocol is also bipartite with respect to any bipartition, it must be that \begin{subequations}
\begin{align}
    E_{sq}^{A|BR}\left(\ket{\psi}^{RAB}\right) \geq E_{sq}^{A|B}\left(\ket{\phi}^{AB}\right), \\
    E_{sq}^{B|AR}\left(\ket{\psi}^{RAB}\right) \geq E_{sq}^{A|B}\left(\ket{\phi}^{AB}\right).
\end{align}
\end{subequations}
This means that the entanglement entropy of $\ket{\phi}^{AB}$ is bounded above by $\min\left\{H\left(\psi^{A}\right),H\left(\psi^{B}\right)\right\}$.

\end{document}